\documentclass{aptpub}

\authornames{Weiyu Xu, A. Kevin Tang} % insert the authors here for use in running head
\shorttitle{A Generalized Coupon Collector Problem} % insert short title here for use in running head

\usepackage{amsmath,amsfonts,graphicx,subfigure,multirow,verbatim}
\usepackage[T1]{fontenc}
\usepackage[latin1]{inputenc}
\usepackage{float}
\floatstyle{ruled}
\newfloat{algorithm}{tbp}{loa}
\floatname{algorithm}{Algorithm}

\def\qed{\mbox{ }~\hfill~$\Box$}

\begin{document}

\title{A Generalized Coupon Collector Problem}

%\author{
%\authorblockN{Weiyu Xu}
%\authorblockA{School of Electrical and Computer Engineering
%\\Cornell University \\
%Ithaca, NY 14853\\
%wx42@cornell.edu}
%\and
%\authorblockN{A. Kevin Tang}
%\authorblockA{School of Electrical and Computer Engineering \\
%Cornell University\\
%Ithaca, NY 14853\\
%atang@ece.cornell.edu}}

\authorone[Cornell University]{Weiyu Xu} % Affiliation is just the name of your university or institution
\addressone{School of Electrical and Computer Engineering, Cornell University, Ithaca, NY 14853, USA. Email:
wx42@cornell.edu.} % Your postal address goes here.

\authortwo[Cornell University]{A. Kevin Tang} % Affiliation is just the name of your university or institution

\addresstwo{School of Electrical and Computer Engineering, Cornell University, Ithaca, NY 14853, USA. Email:
atang@ece.cornell.edu.} % Your postal address goes here.

%\maketitle

\begin{abstract}
This paper provides analysis to a generalized version of the
coupon collector problem, in which the collector gets $d$ coupons
each run and he chooses the one that he has the least so far. In
the asymptotic case when the number of coupons $n$ goes to
infinity, we show that on average $\frac{n\log n}{d} +
\frac{n}{d}(m-1)\log\log{n}+O(mn)$ runs are needed to collect $m$ sets of coupons. An efficient exact
algorithm is also developed for any finite case to compute the
average needed runs exactly. Numerical examples are provided to
verify our theoretical predictions.
\end{abstract}

 % insert keywords separated by a semicolon


\section{Introduction}
The classic coupon collector problem asks for the expected number
of runs to collect a complete set of $n$ different coupons while
during each run the collector randomly gets a coupon. The answer
is $nH_n$ where $H_n=\sum_{k=1}^{n} \frac{1}{k}$ is the harmonic
number \cite{Fel50}. One can further ask for the expected number
of runs to collect $m$ complete sets of coupons, which has been
addressed by Newman and Shepp \cite{NeS60}.

The coupon collector problem and its variants are of traditional
and recurrent interest \cite{FZe03,Kan05,MWi03,Hol01,Nea08}. Besides of their
rich theoretical structures and implications, there are various
applications of them as well including dynamic resource
allocation, hashing and online load balancing \cite{ABK99}, to
just name a few. In particular, we want to point out that these
problems also serve as basic models to analyze delay for
opportunistic scheduling in broadcast wireless fading channels
\cite{ShH06}. For example, to maximize system throughput, we
should serve to the user whose channel condition is the best at
every time slot. In order to evaluate performance, one may ask
about the expected number of time slots needed for all users to be
served at least once. Assuming all channels are i.i.d., then this
is equivalent to the classic coupon collector problem.

In this paper, we investigate a natural generalization of the
coupon collector problem. Instead of getting one coupon, the
collector gets $d$ ($1 \leq d \leq n$) \emph{distinct} coupons randomly each run
and he picks one that so far he has the least. Formally, we denote
the number of runs (a ``run'' often referred to a unit of ``time slot'' or simply a unit of ``time'' in this paper) that are used to collect $m$ sets of coupons as $D^{d}_{m,n}$ and we are interested in characterizing the mean
value of this random variable $D^{d}_{m,n}$ especially in the asymptotic region
when $n$ is large. Clearly, when $d=1$, we go back to the classic
cases; when $d=n$, there is no randomness and
$D^{n}_{m,n}=mn$. In the scheduling transmission context we just
discussed, $d$ can be viewed as a parameter to control the
tradeoff between efficiency (high throughput) and fairness among
all users with $d=1$ purely focusing on efficiency while $d=n$
giving out perfect fairness.

In the remaining part of this paper, we first briefly review
existing related results in Section \ref{sec:existing}. Although they are all
special cases of the general problem, the techniques used to
derive them cannot be applied directly to the general case.
Instead, we develop a new technique to characterize
$E(D^{d}_{m,n})$ and provide upper and lower bounds for $E(D^{d}_{m,n})$ in Section \ref{lower} and Section \ref{sec:upper}.
An asymptotic analysis shows that the upper bound and lower bound match in the asymptotic regime of $n \rightarrow \infty$ in Section \ref{sec:approximation} .
Furthermore, for any finite $n$, an
algorithm is motivated and proposed in Section \ref{sec:accurate} to calculate
$E(D^{d}_{m,n})$ exactly. Finally, we use numerical examples to
validate our theoretical predictions in section
\ref{sec:simulation}.

\section{Existing Results}
\label{sec:existing} The existing results on special cases are
listed below. If $d=1$, then the problem is solved for all $m \geq 1$. If $d>1$, then
only the $m=1$ case is known.

\begin{itemize}
\item $d=1,m=1$ (\cite{Fel50}). It is clear that the number of
runs needed to obtain $(i+1)$-th coupon after obtaining the $i$-th
one follows a geometric distribution with parameter
$\frac{n-i}{n}$. Therefore,
\begin{equation} E(D^{1}_{1,n})=n H_n=n\sum_{k=1}^{n} \frac{1}{k} \end{equation}
For large $n$, \begin{equation} E(D^{1}_{1,n})=n\log n+nO(1)
\end{equation} One sees that the randomness cost is expressed
approximately by a factor $\log n$.

\item $d=1,m \geq 1$ (\cite{NeS60}).
\begin{equation}
E(D^{1}_{m,n})=n \int_0^{\infty}(1-(1-S_m(t)e^{-t})^n)dt
\end{equation} where $S_m(t)=\sum_{k=0}^{m-1}\frac{t^k}{k!}$.

For fixed $m$ and large $n$, \begin{equation} E(D^{1}_{m,n})=n\log n+
n(m-1)\log\log n+ n O(1).
\end{equation}
It is interesting to note that although collecting the first set
of coupons needs about $n\log n$ runs, all later sets only need
$n\log\log n$ runs per set.

\item $d \geq 1, m=1$ (\cite{ShH06}).
This has applications in the scheduling of data packets transmission over wireless channels. Here

\begin{equation} E(D^{d}_{1,n})=\sum_{i=0}^{n-1}\frac{1}{1-\frac{\binom{i}{d}}{\binom{n}{d}}}, \end{equation}
where $\binom{i}{d}=0$ if $i<d$.

For fixed $m$ and large $n$, \begin{equation}
E(D^{d}_{1,n})\sim \frac{1}{d} n\log n.
\end{equation}
This shows that for the $m=1$ case, allowing choosing $d$ coupons
randomly each time decreases the expected number of runs and most
of the decreases occurs from $d=1$ to $d=2$.

\item $d \geq 1, m \geq 1$.
In the mentioned context of scheduling, if a transmitter wants to send $m$ packets to each of the $n$ users, but each time he can only transmits one packet to one user chosen from the $d$ users who have the best wireless communication channels. Due to the time varying nature of the wireless channels, it is natural to assume that for each time index, the $d$ users who have the best communication channels are uniformly distributed among the $n$ users. So $E(D_{m,n}^d)$ gives an estimate on the total delay in delivering these $m$ packets, and, in this paper, we will offer a characterization of $E(D_{m,n}^d)$.

%However, a closed-form solution for the general case $m \geq 1$ and $d\geq 1$ is not known yet.
%n\int_{0}^{1-d/n}\frac{dx}{1-x^d}+O(1)\end{equation}

\end{itemize}

\section{Lower Bound on $E(D_{m,n}^{d})$}\label{lower}
We will first lower bound $E(D_{m,n}^{d})$ by considering a different coupon collecting process. In this new process, each time we uniformly select $d$ distinct coupons out of $n$ coupons, and instead of keeping only one coupon out of these $d$ selected coupons, we would keep all these $d$ coupons. Apparently, the expected time of collecting $m$ sets of coupons in this way will be no larger than the process which only keeps one coupon a time.

However, it is not so straightforward to directly get a estimate for this new process. This motivates us to consider another process in which each time, one will collect $d$ uniformly, independently chosen (allowing repeating) coupons and keep all of them. This process stops when $m$ sets of coupons are fully collected.

\begin{lemma} \label{lem:distinctcloseness}
Let $t_{1}$ be the expected time to collected $m$ sets of coupons for the process in which each time $d$ uniformly chosen distinct coupons are kept. Let $t_{2}$ be the expected time to collected $m$ sets of coupons for the process in which each time $d$ uniformly chosen (allowing repetition) coupons are kept. Then

\begin{equation} \label{eqnud}
t_{1} \geq \frac{\binom{n}{d}}{n^d}  t_{2}.
\end{equation}

\end{lemma}

\begin{proof}
We simulate the process of choosing $d$ distinct coupons through an expurgated process of choosing $d$ independent coupons (allowing repetition). If the $d$ coupons we independently choose (allowing repeating) are not distinct, we will discard this group of $d$ coupons; if they are all distinct $d$ coupons, we will keep them. The kept coupons from expurgated process follow the same distribution as the chosen $d$ distinct coupons. However, the expected time for one to get a group of $d$ distinct coupons is clearly $\frac{n^d}{\binom{n}{d}}$. So in the worst case, $t_2 \leq \frac{n^d}{\binom{n}{d}} t_{1}$. \qed
\end{proof}

In summary, in order to give a lower bound on $E(D_{m,n}^{d})$, we will first need a lower bound on $t_2$ for the process of keeping $d$ uniformly randomly chosen coupons (allowing reptition). To do this, we follow the approach of generating functions in \cite{NeS60}.

Let $p_{i}$ be the probability of failure of obtaining $m$ sets of coupons when we have kept $i$ coupons. Let $P_{x_{1},\cdots, x_{n}}$ be a power series and let $\{P_{x_{1},\cdots, x_{n}}\}$ be the power series when all terms having all exponents $\geq m$ have been removed. By these notations,
\begin{equation*}
t_{2}=\sum_{j=0}^{\infty} p_{dj},
\end{equation*}
and
\begin{equation*}
p_{dj}=\frac{ \{(x_{1}+\cdots+ x_{n})^{dj}\} }{n^{dj}},
\end{equation*}
with ${x_{1},\cdots, x_{n}}$ all equal to $1$.

In addition, we know
\begin{equation*}
E(D^{1}_{m,n})=\sum_{q=0}^{d-1}\sum_{j=0}^{\infty} p_{dj+q}=n \int_0^{\infty}(1-(1-S_m(t)e^{-t})^n)dt,
\end{equation*}
where $S_m(t)=\sum_{k=0}^{m-1}\frac{t^k}{k!}$ \cite{NeS60}.

We also notice that $p_{i}$ is nonincreasing as $i$ grows, so
\begin{eqnarray*}
t_{2}&=&\sum_{j=0}^{\infty} p_{dj}
\geq (\sum_{j=0}^{\infty} p_{j})/d\\
&\geq& n \int_0^{\infty}(1-(1-S_m(t)e^{-t})^n)dt/d.
\end{eqnarray*}
So by (\ref{eqnud}), we know
\begin{equation*}
E(D_{m,n}^d)\geq  \frac{\binom{n}{d}}{n^d}\left( E(D_{m,n}^1)/d \right ).
\end{equation*}
\section{Upper Bound on $E(D_{m,n}^{d})$}\label{sec:upper}
In this section, we will upper bound the expected time of collecting $m$ complete sets of coupons. To achieve this, we will upper bound the expected time for collecting $m$ complete sets of coupons in a suboptimal process. In this new process, each time, we will uniformly and independently choose $d$ coupons (allowing repetition). Among this group of $d$ coupons, we will start looking at them one by one. If the $i$-th  ($1\leq i\leq d$) coupon is the first such a coupon that we so far have fewer than $m$ copies, then we will keep this $i$-th coupon and discard the remaining $(d-i)$ coupons.

First of all, we observe that $d$ distinct coupons are favorable in terms of minimizing the collection time compared with $d$ coupons with possible repeating.

\begin{theorem}\label{thm:bound2}
The minimized expected time of collecting $m$ sets of coupons, when each time the coupon collector is given $d$ uniformly chosen distinct coupons but is only allowed to keep $1$ coupon, is no bigger than  the minimized expected time of collecting $m$ set of coupons, when  each time the coupon collector is given $d$ uniformly chosen coupons (allowing repeating) but is only allowed to keep $1$ coupon.
\end{theorem}

\begin{proof}
Apparently, when the coupon collector is given $d$ distinct coupons, he has more choices in making his decisions.\qed
\end{proof}

Secondly, we show that it is an optimal strategy for the coupon collector to keep the coupon out of the $d$ incoming coupons (whether allowing repetition or not), for which he has the fewest copies.
\begin{theorem}\label{thm:bound1}
The expected time of collecting $m$ sets of coupons is minimized when each time, the coupon collector keeps the coupon which he has the fewest so far, if the coupon collector is allowed to keep only $1$ out of the $d$ offered coupons.
\end{theorem}

\begin{proof}
Suppose (before the coupon collector finishes collecting all $m$ sets of coupons) among the uniformly chosen $d$ (either distinct or allowing repeating) coupons, the $j$-th type of coupon (there are in total $n$ types of coupons and $1\leq j \leq n$) is what he has the fewest, say $c_{1}$ copies. Suppose further that he chooses instead to keep a different type of coupon, say the $l$-th ($1\leq l \leq n$, $l\neq j$) type of coupon, and, for this type of coupon, the coupon collector has already got $c_{2}$ copies, where $c_{2}>c_{1}$. Immediately, we know that after keeping the $l$-th type of coupon, we have at least $c_{2}+1$ copies of $l$-th type of coupons, which satisfies $c_{2}+1\geq c_{1}+2$. We call the resulting state for the kept coupons as $A$, identified by the tuple $(c_1,c_{2}+1)$. Otherwise we would just keep the $j$-th type of coupon, then we will have $c_{1}+1$ coupons of type $j$ and $c_{2}$ coupons of type $l$. We call the resulting state in this case as $B$, identified by the tuple $(c_1+1,c_2)$.

Now we argue that to collect $m$ sets of coupons, on average starting from state $B$ will take no longer time than starting from state $A$. The main idea is to let the collector starting from $B$ follows the strategy of the collector starting from $A$, and do no worse in the expected delay.

%To see this, we will look at the probability $P_{A, t_{3},T}$ ($t_3\geq 1$) that starting from state $A$, where the coupon collector just collects fully the $m$ sets of coupons at time index $t_{3}$. We set this time index $t_3=0$ when we start from state $A$, and increase it by one whenever the coupon collector chooses $d$ coupons. $T$ is the set of time indices where the coupon collector decides to keep one coupon of type $j$ or $l$.

Let us consider two coupon collectors, one starting from state $A$ and the other starting from state $B$. At each time index, these two collectors get the same $d$ coupons. A coupon collector will keep one coupon only if he has fewer than $m$ copies of coupons of that type. Suppose that the coupon collector starting from state $A$ follows his optimized ``keeping'' decision such that his expected time to fully collect $m$ sets of coupons is minimized.
%We now focus on all the possible paths in collecting the coupons.
%Each path is specified by which $d$ coupons are offered at each time index.
% and also by the coupon that the coupon collector decides to keep at each time index.
%According to the optimal decision scheme (in terms of minimizing the expected collection time) for the coupon collector starting from state $A$, the coupon collector starting from state $B$ designs a specific decision scheme.
Then we let the coupon collector starting from state $B$ follow the same decision process as the coupon collector starting from state $A$, until some time index when the coupon collector starting from state $A$ decides to keep a coupon of type $j$ or type $l$. At that time index, we also let the coupon collector starting from state $B$ keep the same type of coupon as the coupon collector starting from state $A$ does. Then let $\{c_{1}', c_{2}'\}$ denote the resulting numbers of kept coupons of type $j$ and type $l$, for the collector starting from state $A$; and likewise define $\{c_{1}'', c_{2}''\}$ for the collector starting from state $B$.

Now we can take an inspection of state $\{c_{1}', c_{2}'\}$ and state $\{c_{1}'', c_{2}''\}$. There are two scenarios to discuss separately.

In the first scenario, $c_{2}=c_{1}+1$ and it is the coupon of type $j$ that the coupon collector decides to keep at that time index. Then $\{c_{1}', c_{2}'\}=\{c_{1}+1, c_{2}+1\}$ and $\{c_{1}'', c_{2}''\}=\{c_{1}+2, c_{2}\}=\{c_{2}+1, c_{1}+1\}$. Apparently $\{c_{1}', c_{2}'\}$ and $\{c_{1}'', c_{2}''\}$ are just the permutations of each other, and so by symmetry, the optimized time from these two new states to the completion of collecting $m$ sets of coupons will be the same.

In the second scenario, we have $-c_{1}'+c_{2}'>-c_{1}''+ c_{2}''\geq 0$. In this scenario, we update state $A$ as $\{c_{1}', c_{2}'\}$ and update state $B$ as $\{c_{1}'', c_{2}''\}$; and then construct an iterative process of evolving states $A$ and $B$ as follows. (To keep the notations consistent, at the beginning of a new iteration, we always represent state $A$ and state $B$ by $\{c_{1}', c_{2}'\}$  and $\{c_{1}'', c_{2}''\}$, even though they are different numbers than in previous iterations.  We also remark that during these iterations, $c_{1}',c_{2}',c_{1}'',c_{2}''$ always satisfy the listed constraints (\ref{eq:cnst})-(\ref{eq:cnst2}), which will be obvious from the iterative process description.)

%At the beginning of each iteration, we always let $\{c_{1}', c_{2}'\}$ (for state $A$) and $\{c_{1}'', c_{2}''\}$ (for state $B$), representing how many type $j$ and type $l$ each state has respectively.
 At the beginning of each iteration, we have two coupon collectors starting from states $A$ and $B$ respectively. At each time index, the coupon collector starting from state $B$ will keep the same type of coupon as the coupon collector starting from state $A$, until some run they keep a coupon of either type $j$ or type $l$. Again, we have two cases to consider.

In the first case, $c_{1}''=c_{2}''$ and it is the coupon of type $j$ that the coupon collectors decide to keep at that time index. After keeping that
coupon, we have states $\{c_{1}'+1, c_{2}'\}$ and $\{c_{1}''+1, c_{2}''\}$ for the collectors starting from $A$ and starting from $B$, respectively.
%By symmetry,  state $\{c_{2}'', c_{1}''+1\}$ is equivalent to state $\{c_{1}''+1, c_{2}''\}$.
So if $c_{1}''+1=c_{2}'$, then by symmetry, $\{c_{1}''+1,c_{2}''\}$ and $\{c_{1}'+1, c_{2}'\}$ are two equivalent states, and we are done (note that $c_{1}'+c_{2}'=c_{1}''+ c_{2}''$)
; otherwise, if $c_{1}''+1<c_{2}'$, we update state $A$ and $B$ as $\{c_{1}'+1, c_{2}'\}$ and $\{c_{2}'', c_{1}''+1\}$ respectively.% without affecting the expected coupon collection time by exchanging the two elements in state $B$.

In the second case, after keeping that new coupon of type $j$ or type $l$, we simply update state $A$ and state $B$ respectively to record the new numbers of type $j$ and type $l$ coupons for these two coupon collectors.  Then we go back to the beginning of another iteration.

%, which represent respectively the number of coupons of type $j$ and type $l$ in each state.

Note in all these iterations, we always maintain
\begin{eqnarray}
\label{eq:cnst}
c_{1}'+c_{2}'&=&c_{1}''+ c_{2}'',\\
c_{1}'&\leq &c_{2}',\\
c_{1}''&\leq &c_{2}'',\\
c_{2}'-c_{1}'&>&c_{2}''-c_{1}'',\\
c_{1}',c_{2}',c_{1}'',c_{2}''&\leq& m. \label{eq:cnst2}
\end{eqnarray}
Because of this, in each iteration, when the coupon collector starting from $A$ can keep a certain type of coupon, the coupon collector starting from state $B$ can also keep the same type of coupon.
%So it is always for the coupon collector starting from state $B$ to follow the decision policy for the coupon collector starting from $A$.

Because for every iteration, we will increase $c_{1}'+c_{2}'$ and $c_{1}''+ c_{2}''$, and $c_{1}',c_{2}',c_{1}'',c_{2}''\leq m$, if we iterate the previous processes, we will eventually run into a pair of symmetric states in some iteration for these two coupon collectors. % for $\{c_{1}', c_{2}'\}$ and $\{c_{1}'', c_{2}''\}$, when the decision process for state $B$ will have the same expected collection time as the decision process for state $A$.
%
%For example, they can in the end get to the state of $\{m-1,m\}$ for both $\{c_{1}', c_{2}'\}$ and $\{c_{1}'', c_{2}''\}$.

%Also let us fix $t_{3}$ and $T$ and consider all the possible paths that the coupon collector just collects $m$ sets of coupons at time $t_{3}$ and collects coupons of type $j$ or $l$ only at the time indices from $T$. Each path is specified by the coupons randomly chosen at each time index and also by the coupon that the coupon collector decides to keep for each time index.

%Now we devise a possibly suboptimal decision process for the coupon collector if he starts instead from state $B$ at $t_{3}=0$.
%
%
%Now the question is whether to fill up the support time indices $T$, subject to the ``policy'' constraint (namely we are adopting a policy, not that we can use any posterior information from later happening), whether there is a high probability for filling up $T$ for state $B$ than $A$.
Since we can always end up in a symmetric state for the two coupon collectors, starting from $B$ to collect $m$ sets of coupons will not take longer time than starting from $A$.

\qed
\end{proof}

In fact, the previous arguments can essentially show that starting from a state $\{c_{1}', c_{2}'\}$ has no bigger expected collection time than from the state $\{c_{1}'', c_{2}''\}$, if $|c_{1}'- c_{2}'|<|c_{1}''-c_{2}''|$ and $c_{1}'+c_{2}'=c_{1}''+ c_{2}''$.

At this point, we are ready to present the following upper bound for $E(D_{m,n}^{d})$.

\begin{theorem}
Suppose that the coupon collector is given $d$ uniformly randomly chosen $d$ distinct coupons and he is only allowed to keep one out of these $d$ distinct coupons. Then the expected time  $E(D_{m,n}^{d}) \leq \frac{E(D_{m,n}^{1})}{d}+mn(1-1/d)$.
\end{theorem}

\begin{proof}
From Theorem \ref{thm:bound1} and \ref{thm:bound2}, we will consider an upper bound on the expected finishing time if each run the coupon collector is given $d$ independently chosen coupons (allowing repeating) and he decides to keep only the first ``useful'' coupon when it is available among the $d$ coupons. Here a kept ``useful'' coupon means this coupon is kept before the coupon collector has $m$ copies of that type of coupon.

The idea of the proof is to upper bound the expected finishing time conditioning on a a specific sequence of kept coupons, until we have $m$ sets of coupons. By a specific sequence of ``keeper'' coupons, we mean a sequence of kept ``useful'' coupons, specified in their types and the order of keeping them. We note that a specific sequence of ``keeper'' coupons contain $mn$ coupons.

First we make a key observation about the probability that the coupon collector follows a specific sequence of ``keeper'' coupons. This probability will be the same as the corresponding probability that he follows this specific sequence of kept coupons, in another collection process where each run he is only offered $1$ instead of $d$ independently, uniformly chosen coupons. This is because, when the coupon collector is offered $d$ coupons, he still checks them one by one and only keeps the first one that is ``useful''.

Now, conditioning on a specific sequence of kept coupons, we are interested in the expected time to collect that sequence. Suppose that right after the $r$-th coupon in this ``keeper'' sequence has just been kept, there are $s$ types of coupons for which the coupon collector have $m$ copies. We then want to know what is the expected number of runs needed to collect the next $(r+1)$-th ``keeper'' coupon, conditioning on the whole ``keeper'' coupon sequence is known.

Through the conditional event that the ``keeper'' sequence has already specified, it is a standard exercise to show that the average time the coupon collector takes to collect the $(r+1)$-th ``keeper'' coupon will be
\begin{equation*}
E'=\frac{1}{1-(1-\frac{n-s}{n})^d}.
\end{equation*}

In fact, conditioned on the fact that the next ``keeper'' coupon is already specified, unless the next inspected coupon belongs to the $s$ types of coupons for which the collector has already had $m$ copies, it must be the $(r+1)$-th ``keeper'' coupon. So conditioned on the already specified next ``keeper'' (the same as conditioned on the whole ``keeper'' sequence is known, because the collecting process is a Markov chain), a uniformly chosen coupon is the specified ``keeper'' with probability $\frac{n-s}{n}$. So conditioned on the $(r+1)$-th ``keeper'' is known, with probability $(1-\frac{n-s}{n})^d$,  none of the $d$ uniformly chosen (allowing repeating) coupons is the known $(r+1)$-th ``keeper''.

However, we note that if the coupon collector is only offered $1$ instead of $d$ coupons each time, then the expected time to get the $(r+1)$-th coupon in that sequence on average will be
\begin{equation*}
E''=\frac{n}{n-s}.
\end{equation*}

By Lemma \ref{lem:asymptotics0} proven latter than, we know that
\begin{equation*}
E'-E''/d \leq 1-\frac{1}{d}
\end{equation*}
for any $1\leq s \leq n$.

Since there are exactly $mn$ coupons in any specific sequence of ``keeper'' coupons, the total expected time $E_{S}'$ for collecting a whole ``keeper'' sequence $S$, when each time the coupon collector has $d$ coupons (allowing repeating) to choose from, and the total expected time $E_{S}''$ for collecting the same ``keeper'' sequence $S$, when each time the coupon collector only has $1$ incoming random coupon, satisfy
\begin{equation} \label{eqn5}
E_{S}'-E_{S}''/d \leq mn(1-\frac{1}{d}).
\end{equation}

By invoking the fact that the coupon collector follows any specific sequence of ``keeper'' coupons with the same probability for both the $d=1$ case and the $d\neq 1$, (\ref{eqn5}) implies
\begin{equation}
E(D_{m,n}^{d}) \leq \frac{E(D_{m,n}^{1})}{d}+mn(1-1/d),
\end{equation}
for any $d$, which is exactly the theorem statement.\qed
\end{proof}

\begin{lemma}\label{lem:asymptotics0}
Function $f(i) = \frac{n}{d\cdot i} -\frac{1}{1-(1-\frac{i}{n})^d}$
is decreasing for $1\leq i\leq n$; and, $\frac{1}{d}-1\leq f(i)\leq 0$ for $1\leq i\leq n$.
\end{lemma}

\noindent {\bf Proof.}
We need to show that the derivative
\[f'(i) = \frac{d (1-\frac{i}{n})^{d-1}}{n (1-(1-\frac{i}{n})^d)^2}-\frac{n}{d\cdot i^2}\]
for $i \in [1,n]$.

Let \[g(x) = (1 - x)^d + d\cdot x (1 - x)^{\frac{d - 1}{2}}.\] So
\[g'(x) = -\frac{1}{2} d (1-x)^{\frac{d-3}{2}}\left(x+ d\cdot
x+2(1-x)^{\frac{d+1}{2}}-2 \right)\]

Also let \[h(x) = x+ d\cdot x+2(1-x)^{\frac{d+1}{2}}-2.\]
So \[h'(x) = 1+d-(1+d)(1-x)^{\frac{d-1}{2}} \geq 0,\]
and  $h(x)\geq h(0) = 0$ for $x>0$.
Because
\[g'(x) = -\frac{1}{2} d (1-x)^{\frac{d-3}{2}}h(x)\leq 0 \mbox{\; for\; }
0<x\leq 1,\]
we have
\[g(x) \leq g(0) = 1 \mbox{\; for\; }
0<x\leq 1.\]
This translates into
\[1- (1 - x)^d \geq d\cdot x (1 - x)^{\frac{d - 1}{2}}\mbox{\; for\; }
0<x\leq 1,\]
so
\[\frac{1}{(1- (1 - x)^d)^2} \leq \frac{1}{d^2\cdot x^2 (1 - x)^{d-1}}\mbox{\; for\; }
0<x\leq 1.\]
Namely
\[\frac{d(1 - x)^{d-1}}{n(1- (1 - x)^d)^2} \leq \frac{1}{n\cdot d\cdot x^2 }\mbox{\; for\; }
0<x\leq 1.\]
Plugging in $x=\frac{i}{n}$, we have
\[\frac{d(1 - \frac{i}{n})^{d-1}}{n(1- (1 - \frac{i}{n})^d)^2} \leq \frac{n}{d\cdot i^2 }\mbox{\; for\; }
1\leq i\leq n.\]
So
\[f'(i)\leq 0 \mbox{\; for\; }1\leq i\leq n.\]
Calculating $f(n)$, we have
\[\frac{n}{d\cdot i} -\frac{1}{1-(1-\frac{i}{n})^d}\geq \frac{1}{d}-1\mbox{\; for\; }1\leq i\leq n. \]\qed

\section{An Asymptotic Analysis ($n\rightarrow \infty$)}
\label{sec:approximation}

In this section, we will give an asymptotic analysis of the upper and lower bounds for $E\{D^{d}_{m,n}\}$ and see how it behaves asymptotically for fixed $d$ and $m$ as $n$ goes to $\infty$. We will begin with an asymptotic analysis through an exact expression for $E\{D^{d}_{1,n}\}$.

\begin{theorem}\label{thm:asymptotics0}
When $n$ is large enough and $d>1$,
\begin{equation}
E[D^{d}_{1,n}]=n\left(\frac{\log n}{d}+O(1)\right)
\end{equation}\end{theorem}

\noindent {\bf Proof.}
\begin{eqnarray*}
& & E[D^{d}_{1,n}]\\
& = & \sum_{i=0}^{n-1}\frac{1}{1-\frac{{i \choose d}}{{n \choose d}}} = \sum_{i=0}^{n-1}\frac{1}{1-\frac{i(i-1)\ldots(i-d+1)}{n(n-1)\ldots(n-d+1)}}\\
& \geq & d+\sum_{i=d}^{n-1}\frac{1}{1-\left(\frac{i-d+1}{n-d+1}\right)^d}\\
%& = & \sum_{i=1}^{n}\frac{1}{1-\left(\frac{n-d+1-i}{n-d+1}\right)^d} =  \sum_{i=1}^{n}\frac{1}{1-\left(1-\frac{i}{n-d+1}\right)^d}\\
& = & d+\sum_{i=1}^{n-d}\frac{1}{1-\left(1-\frac{i}{n-d+1}\right)^d}.
%& = & \sum_{i=1}^{n}\frac{1}{1-\sum_{j=0}^d {d \choose j} (-\frac{i}{n-d+1})^j}\\
%& = & \sum_{i=1}^{n}\frac{1}{-\sum_{j=1}^d {d \choose j} (-\frac{i}{n-d+1})^j}
\end{eqnarray*}

Since $(1-x)^d \geq 1- dx$ for $1 \leq x \leq 1$,
\begin{eqnarray*}
E(D^{d}_{1,n})%& = & \sum_{i=1}^{n}\frac{1}{\frac{d\cdot i}{n-d+1}-\displaystyle\sum_{j=2}^{d} {d \choose j} \left(-\frac{i}{n-d+1}\right)^j}\\
& \geq & \sum_{i=1}^{n}\frac{n-d+1}{d\cdot i} = \frac{n-d+1}{nd}\sum_{i=0}^{n-1}\frac{n}{n-i}\\
%& = & \frac{n-d+1}{nd}E(D^{1}_{1,n})\\
& = & (1-\frac{d-1}{n})\frac{E(D^{1}_{1,n})}{d} \rightarrow  \frac{1}{d}E(D^{1}_{1,n})
\end{eqnarray*}
as $n\rightarrow \infty$.

\begin{eqnarray*}
E(D^{d}_{1,n})& = & \sum_{i=0}^{n-1}\frac{1}{1-\frac{i(i-1)\ldots(i-d+1)}{n(n-1)\ldots(n-d+1)}}\\
& \leq &\sum_{i=0}^{n-1}\frac{1}{1-(\frac{i}{n})^d}
%& = &\sum_{i=0}^{n-1}\frac{1}{1-(\frac{i}{n})^d}\\
 = \sum_{i=1}^{n}\frac{1}{1-(1-\frac{i}{n})^d}\\
& \leq & \sum_{i=1}^{n}(\frac{n}{d\cdot i}-\frac{1}{d}+1)\\
& = & \frac{1}{d}E(D^{1}_{1,n})+n(1-\frac{1}{d}).
\end{eqnarray*}
\qed

%\begin{theorem}
%\label{thm:asymptotics1} When $n$ is fixed and $d>1$,
%\begin{equation}
%\lim_{m \rightarrow \infty} \frac{E(D^{d}_{m,n})}{nm}=1
%\end{equation}
%\end{theorem}
%
%\noindent {\bf Proof.}
%
%By the law of large numbers,
%\begin{eqnarray*}
%\lim_{m \rightarrow \infty} E(D_{m,n}^{1}) = nm
%\end{eqnarray*}
\begin{theorem}
\label{thm:asymptoticsnlarge} When $m$ is fixed, then for any $d>1$,
\begin{equation}
\lim_{n \rightarrow \infty} \frac{E(D^{d}_{m,n})-n\log(n)/d}{(n(m-1)\log\log n)/d}=1
\end{equation}
\end{theorem}

\noindent {\bf Proof.}
From the lower bound and upper bound for $E(D^{d}_{m,n})$ in Section \ref{lower} and Section \ref{sec:upper}, we know
\begin{equation*}
\lim_{n \rightarrow \infty} \frac{E(D^{d}_{m,n})}{E(D^{1}_{m,n})}=\frac{1}{d}.
\end{equation*}
Then the asymptotic expression emerges immediately by recalling the asymptotic expression for $E(D^{1}_{m,n})$.
\qed

\section{An Algorithmic Approach (for any finite $n$) }
%\section{A Sate-Space Representation for the Generalized Coupon Collector Problem}
In this section, we will give an algorithm which calculates exactly $E(D_{m,n}^{d})$ for specified $m$, $n$ and $d$ based on a state-space representation of the Markov process of collecting the coupons.
%
%Before we present the result for $E(D^{d}_{mn})$ when $n$ is large, we will
%study some numerical examples to gain more insight into the problem.
For each $n_0,n_1,n_2,\cdots, n_m \geq 0$ satisfying $n_0+n_1+\cdots + n_m = n$, define $S_m=(n_0,n_1,\cdots, n_m)$ to be the state where $n_{i}$ ($0\leq i\leq m$) is the number of coupons that the coupon collector has collected $i$ times. Hence, $E(D^{d}_{mn})$ is the expected number of runs for the coupon collector to go from state $(n,0,\ldots,0)$ to state $(0,\ldots,0,n)$. %The following diagram gives the state transition diagram for the case of $n=5,m=2,d=1$.

\label{sec:accurate} We now provide an algorithm to calculate
$E(D^{d}_{m,n})$. Define
\begin{eqnarray*}
& & N_m^d(S_m)\\
 & = & \mbox{starting from state $S_m$, the number of runs after}\\
&& \mbox{which $m$ completed sets of coupons have been}\\
&& \mbox{collected, i.e., the number of runs from state $S_m$}\\
&& \mbox{to $(0,\cdots,0,n)$}
\end{eqnarray*}
Clearly,
\begin{eqnarray}
&& N_m^d(n,0,\cdots, 0)= D_{m,n}^d,\label{eqn1}\\
&&N_m^d(0,\cdots,0, n)= 0,\label{eqn2}
\end{eqnarray}

Suppose we are at state $S_m=(n_0,n_1,\cdots, n_m)$. After one run,
the transition probability from $S_m$ to the following two states are as follows (w.p. is abbreviation for ``with probability''):
\[
\left\{\begin{array}{ll}
(n_0,n_1,\cdots, n_m)& \;\;\, \text{w.p.}\;\; \frac{{{n_m} \choose d}}{{n \choose d}}\\
(n_0,\cdots,n_i-1,n_{i+1}+1,\cdots, n_m)& \;\begin{array}{l}\text{w.p.}
\;\; p_i\\0\leq i<m\end{array}
\end{array}\right.
\]
where $p_i= \left({{\sum_{t=i}^m n_t} \choose d}-{{\sum_{t=i+1}^m
n_t} \choose d}\right)/{n \choose d}$.

Therefore, we have the equation
\begin{eqnarray*}
&& E[N_m^d(n_0,\cdots, n_m)]\\
&=&1+ {{n_m} \choose d}/{n \choose d} \times E[N_m^d(n_0,\cdots, n_m)]\\
&+&\sum\limits_{i=0}^{m-1}p_i
E[N_m^d(n_0,\cdots,n_i-1,n_{i+1}+1,\cdots, n_m)].
\end{eqnarray*}
So
\begin{eqnarray}
&& E[N_m^d(n_0,\cdots, n_m)]\nonumber\\
&=& \frac{{n \choose d}}{{n \choose d}-{{n_m} \choose d}} \nonumber\\
&&\sum\limits_{i=0}^{m-1}\left(
\frac{{{\sum_{t=i}^m n_t} \choose d}-{{\sum_{t=i+1}^m
n_t} \choose d}}{{n \choose d}-{{n_m} \choose d}} \right)\label{eqn4}\\
&& \times E[N_m^d(n_0,\cdots,n_i-1,n_{i+1}+1,\cdots,
n_m)].\nonumber
\end{eqnarray}

Define map $\Phi: \{(n_0,\cdots, n_m): n_0,n_1,\cdots, n_m \geq
0,n_0+n_1+\cdots + n_m = n\}\rightarrow \mathbb{N}$, where
\[\Phi(n_0,n_1,n_2,\cdots, n_m)=\sum_{i=0}^m
(1+n)^{m-i} n_i.\] Obviously, $\Phi$ is an injection and
\begin{eqnarray*}
\Phi(n,0,\cdots, 0)& = & n\cdot(1+n)^m\\
\Phi(0,\cdots, 0,n) & =& n.
\end{eqnarray*}
Since
\begin{eqnarray*}
&&\Phi(n_0,\cdots, n_m)\\
&&~~~~~~-\Phi(n_0,\cdots,n_i-1,n_{i+1}+1,\cdots,
n_m)\\
&=&((1+n)^{m-i} n_i + (1+n)^{m-i-1} n_{i+1})\\
%& & \;\;\;\;\;\; -((1+n)^{m-i} (n_i-1)
&&-((1+n)^{m-i} (n_i-1)
+ (1+n)^{m-i-1} (n_{i+1}+1))\\
&= &(1+n)^{m-i}-(1+n)^{m-i-1}\\
&>& 0,
\end{eqnarray*}
by (\ref{eqn4}), the expected number of runs from a
state $S$ only depends on the expected number of runs from states
$S^*$'s with $\Phi(S^*)<\Phi(S)$. Therefore, we can order all the
states $(n_0,\cdots, n_m)$ according to the value of
$\Phi(n_0,\cdots, n_m)$ and compute $E[N_m^d(n_0,\cdots, n_m)]$
one by one, from the starting state $(0,\cdots, 0,n)$ to the last
state $(n,0,\cdots, 0)$. The algorithm is described in Algorithm
\ref{algorithm 1}.
\begin{algorithm}
\[\begin{array}{l}
\text{\textbf{for} }n_0=0\text{ \textbf{to} }n\\
\; \; \; \text{\textbf{for} }n_1=0\text{ \textbf{to} }n-n_0\\
\; \; \; \; \; \; \cdots \; \; \; \cdots \; \; \; \cdots\\
\; \; \; \; \; \; \; \; \; \text{\textbf{for} }n_{m-1}=0\text{
\textbf{to} }n-\sum_{i=0}^{m-2} n_i\\
\; \; \; \; \; \; \; \; \; \; \; \; \text{\textbf{do}
}n_m=n-\sum_{i=0}^{m-1} n_i\\
\; \; \; \; \; \; \; \; \; \; \; \; \; \; \; \; \; \text{\textbf{if} }n_m=n\\
\; \; \; \; \; \; \; \; \; \; \; \; \; \; \; \; \; \; \; \; \; \text{\textbf{then } }E[N_m^d(n_0,\cdots, n_m)]=0\\
\; \; \; \; \; \; \; \; \; \; \; \; \; \; \; \; \; \; \; \; \;
\text{\textbf{else} \ \ use equation (\ref{eqn4}) to compute}\\
\; \; \; \; \; \; \; \; \; \; \; \; \; \; \; \; \; \; \; \; \; \; \;
\; \; \; \; \; \; \; \ E[N_m^d(n_0,\cdots, n_m)]
\end{array}\]
\caption{Calculating $E(D^{d}_{m,n})$ \label{algorithm 1}}
\end{algorithm}

Since the number of non-negative integer solutions to the equation
$n_0+\cdots+n_m=n$ is ${n+m \choose n}$, the number of states is
${n+m \choose n}$, and the complexity of Algorithm \ref{algorithm 1}
is $O({n+m \choose n})$.

To conclude this section, we now use a simple example ($n=6, m=2$)
to illustrate Algorithm \ref{algorithm 1}. When $m=2$, each state
has $3$ parameters $n_0$, $n_1$, and $n_2$. Since $n_0+n_1+n_2=n$,
we could draw the state transition diagram as in Figure
\ref{fig10}.
\begin{figure}[htbp]
\centering
\includegraphics[bb=0 130 450 550,scale=0.38,angle=-90]{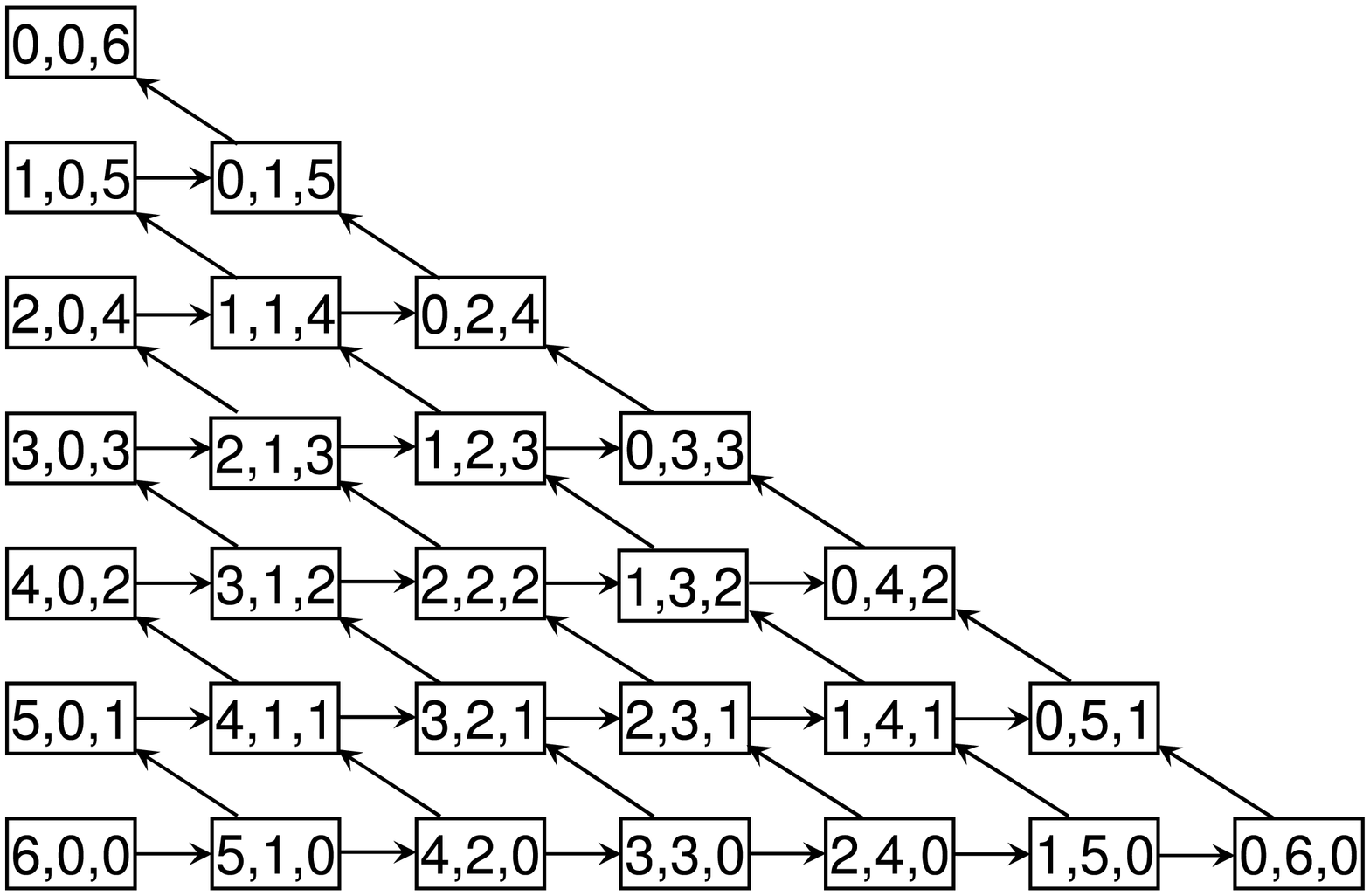}
\caption{State transition diagram for $n=6,m=2$.  Self-loops are
omitted. The nodes are labelled with the value of $n_0,n_1,n_2$.}
\label{fig10}
\end{figure}
Algorithm \ref{algorithm 1} computes $E[N_2^d(n_0,n_1,n_2)]$ for
each state $(n_0,n_1,n_2)$ by the order shown in Figure \ref{fig20}.
This is right because
\begin{itemize}
\item The expected number of runs from any state only depends on the
number of runs from its descents in Figure \ref{fig10}.
\item The computation of $E[N_2^d(\cdots)]$ for any state is done after the computations
for its descents by Figure \ref{fig20}.
\end{itemize}
\begin{figure}[htbp]
\centering
\includegraphics[bb=0 130 450 550,scale=0.38,angle=-90]{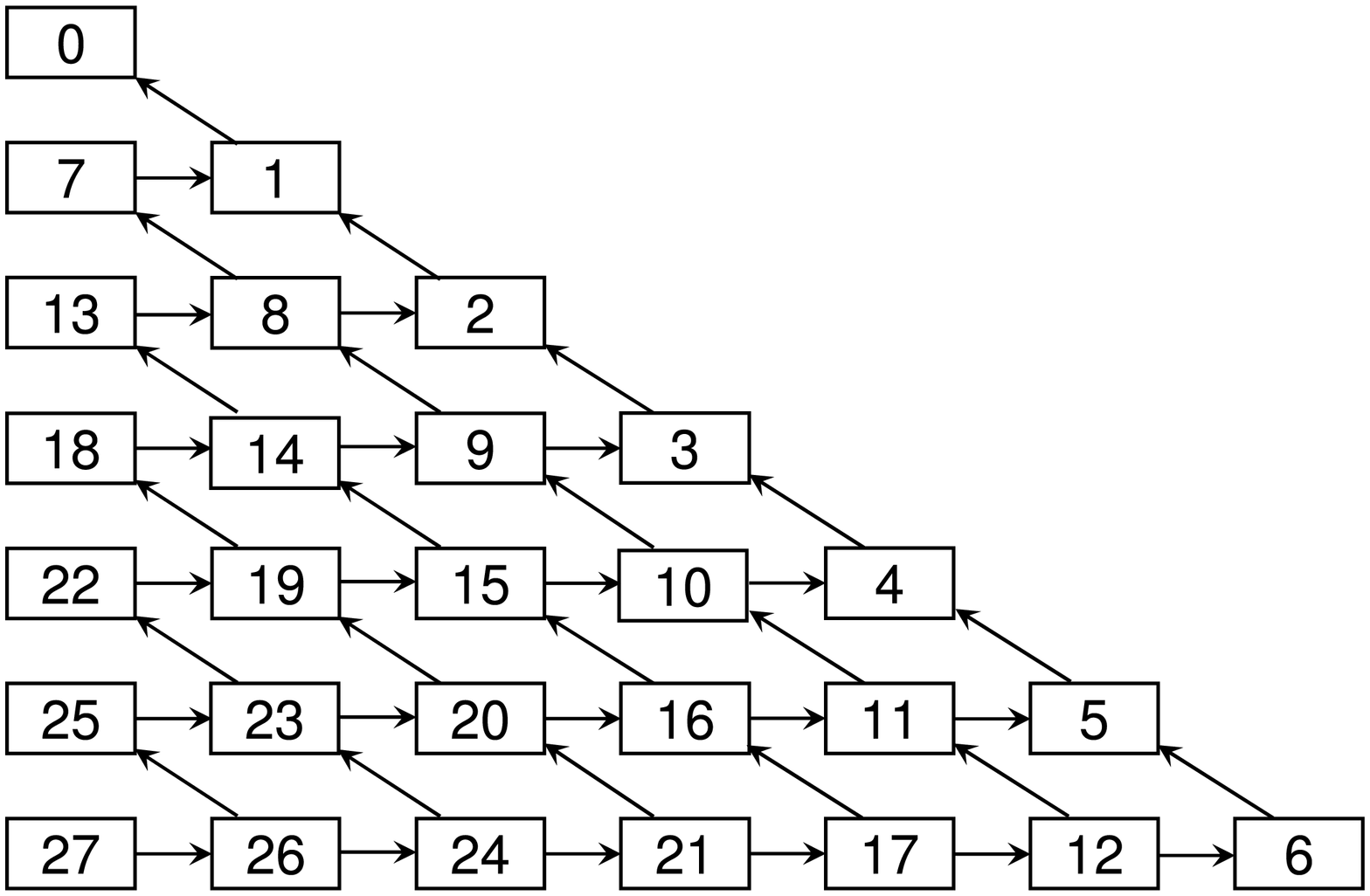}
\caption{State transition diagram for $n=6,m=2$.  The nodes are
labelled by the computation order. The highest node, which
represents $(0,0,6)$ is labelled 0 because $N_2^d(0,0,6)$ is known
to be 0.} \label{fig20}
\end{figure}
The values of $E[N_2^d(n_0,n_1,n_2)]$ for each state $(n_0,n_1,n_2)$
is shown in Figure \ref{fig30}.
\begin{figure}[htbp]
\centering
\includegraphics[bb=0 130 450 550,scale=0.38,angle=-90]{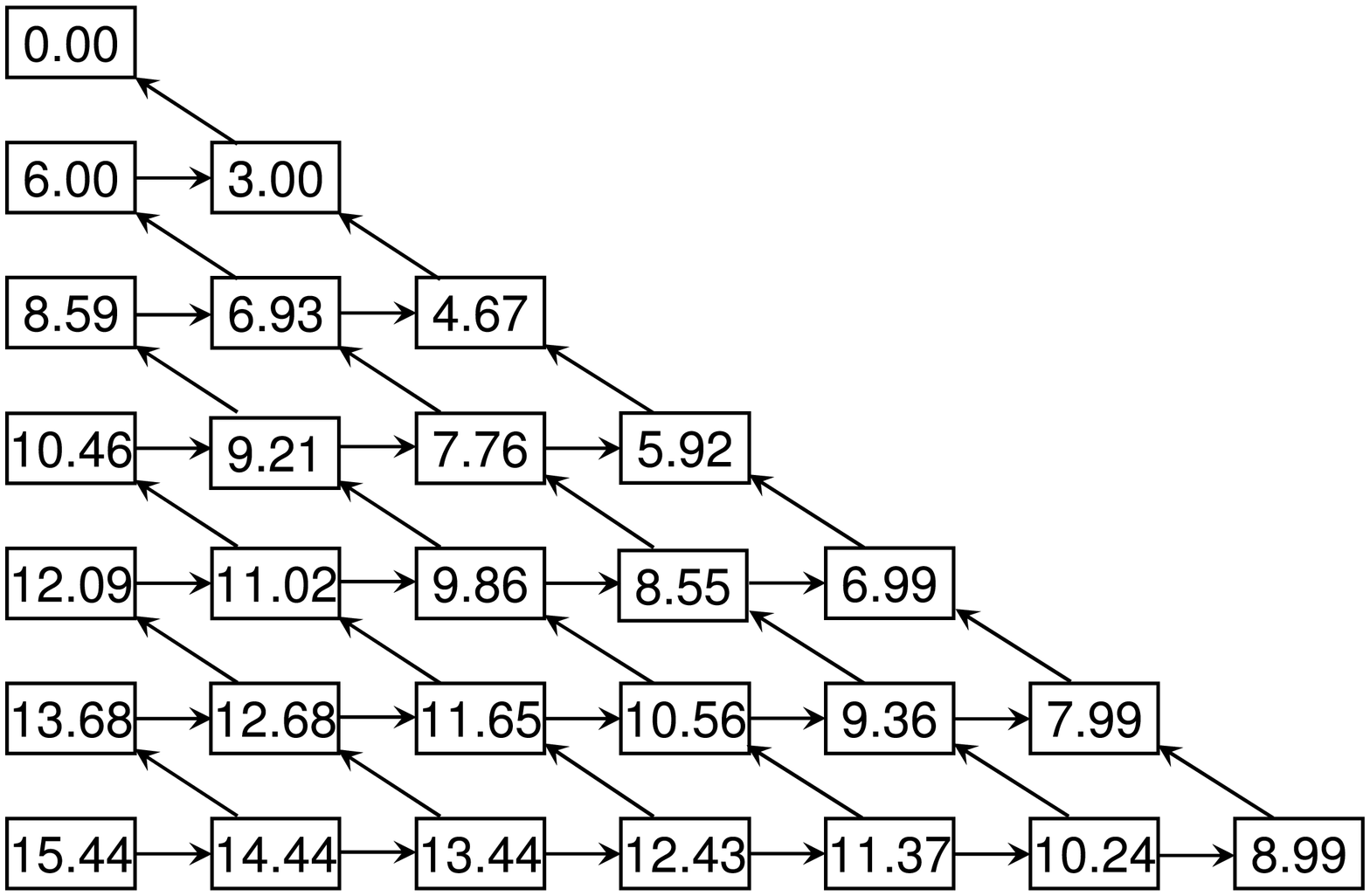}
\caption{State transition diagram for $n=6,m=2$.  The nodes are
labelled by $E[N_2^2(n_0,n_1,n_2)]$.} \label{fig30}
\end{figure}

\section{Numerical Examples}
\label{sec:simulation} We now engage in numerical exercises to
support results in the last two sections, i.e. correctness of
Algorithm \ref{algorithm 1} and the derived upper and lower bounds on $E(D_{m,n}^d)$.

\subsection{Algorithm}
First, we give numerical results for the expected collection time when $n=100$ and $m=1,2,3$ respectively, in three tables. The results show that Algorithm \ref{algorithm 1} gives an expected
delay consistent with the simulation results.

\begin{table}[tbh]
\centering \caption{$m=1$,$n=100$} \label{tab:table1}
\begin{tabular}{|c|c|c|c|c|c|}
\hline \multicolumn{1}{|c|}{$d$} &
\multicolumn{1}{|c|}{1} & \multicolumn{1}{|c|}{2}& \multicolumn{1}{|c|}{3}& \multicolumn{1}{|c|}{4}& \multicolumn{1}{|c|}{5} \\
 \hline
Algorithm & 518.74& 292.93& 220.06&184.79& 164.27 \\
Simulation &518.69& 292.40 & 219.33 & 184.59& 164.18 \\
\hline
\end{tabular}
\end{table}

\begin{table}[tbh]
\centering \caption{$m=2$,$n=100$} \label{tab:table2}
\begin{tabular}{|c|c|c|c|c|c|}
\hline \multicolumn{1}{|c|}{$d$} &
\multicolumn{1}{|c|}{1} & \multicolumn{1}{|c|}{2}& \multicolumn{1}{|c|}{3}& \multicolumn{1}{|c|}{4}& \multicolumn{1}{|c|}{5} \\
 \hline
Algorithm & 728.81& 418.69& 327.02&286.75& 264.84 \\
Simulation &728.20& 419.13 & 327.29 & 286.81& 264.68 \\
\hline
\end{tabular}
\end{table}

\begin{table}[tbh]
\centering \caption{$m=3$,$n=100$} \label{tab:table3}
\begin{tabular}{|c|c|c|c|c|c|}
\hline \multicolumn{1}{|c|}{$d$} &
\multicolumn{1}{|c|}{1} & \multicolumn{1}{|c|}{2}& \multicolumn{1}{|c|}{3}& \multicolumn{1}{|c|}{4}& \multicolumn{1}{|c|}{5} \\
 \hline
Algorithm & 910.87& 531.34&428.75& 386.97& 364.86 \\
Simulation & 910.09& 531.33 & 428.72 & 386.65& 364.90 \\
\hline
\end{tabular}
\end{table}

\newpage
\subsection{Asymptotic Results}
Two cases are considered: $(d=3, m=1)$ and $(d=3, m=2)$. For each
case, three lines are plotted. First, the lower bound from Theorem \ref{thm:bound1}
is plotted. Second, the upper bound is from Theorem \ref{thm:bound2} is computed. Finally, both plots are
compared against the result computed from the algorithm for $n$ from $100$ to $500$.
The results show that the upper bound and lower bound bound the expected collecting time very well. In fact, when $m$ and $d$ are fixed, the upper bound and lower bound will both scale as $\frac{1}{d}E(D_{m,n}^1)$ as $n \rightarrow \infty$.

\begin{figure}[htbp]
\centering
\includegraphics[scale=0.26]{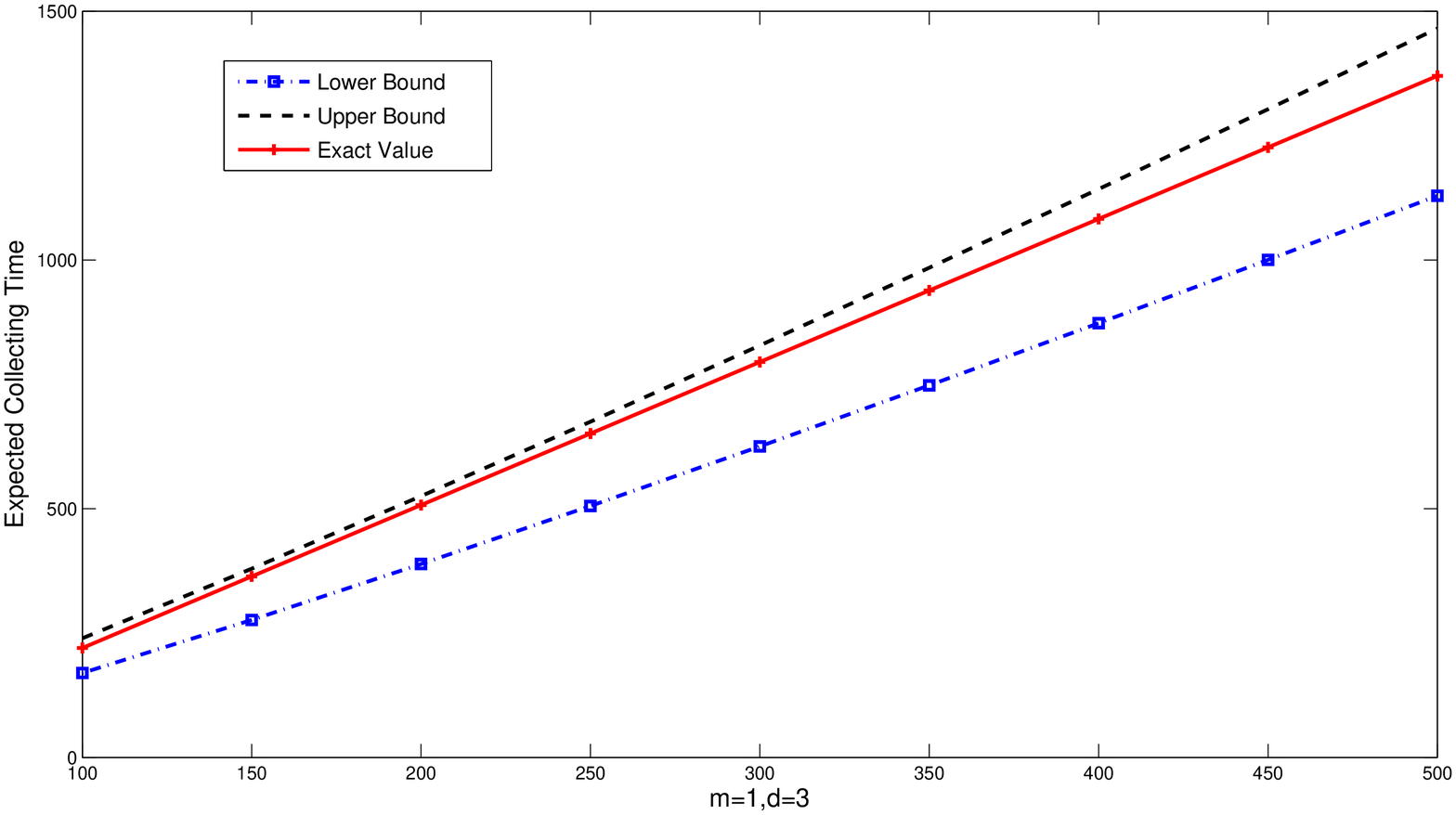}
\caption{Asymptotics for $m=1$, $d=3$.} \label{fig4}
\end{figure}
\begin{figure}[htbp]
\centering
\includegraphics[scale=0.26]{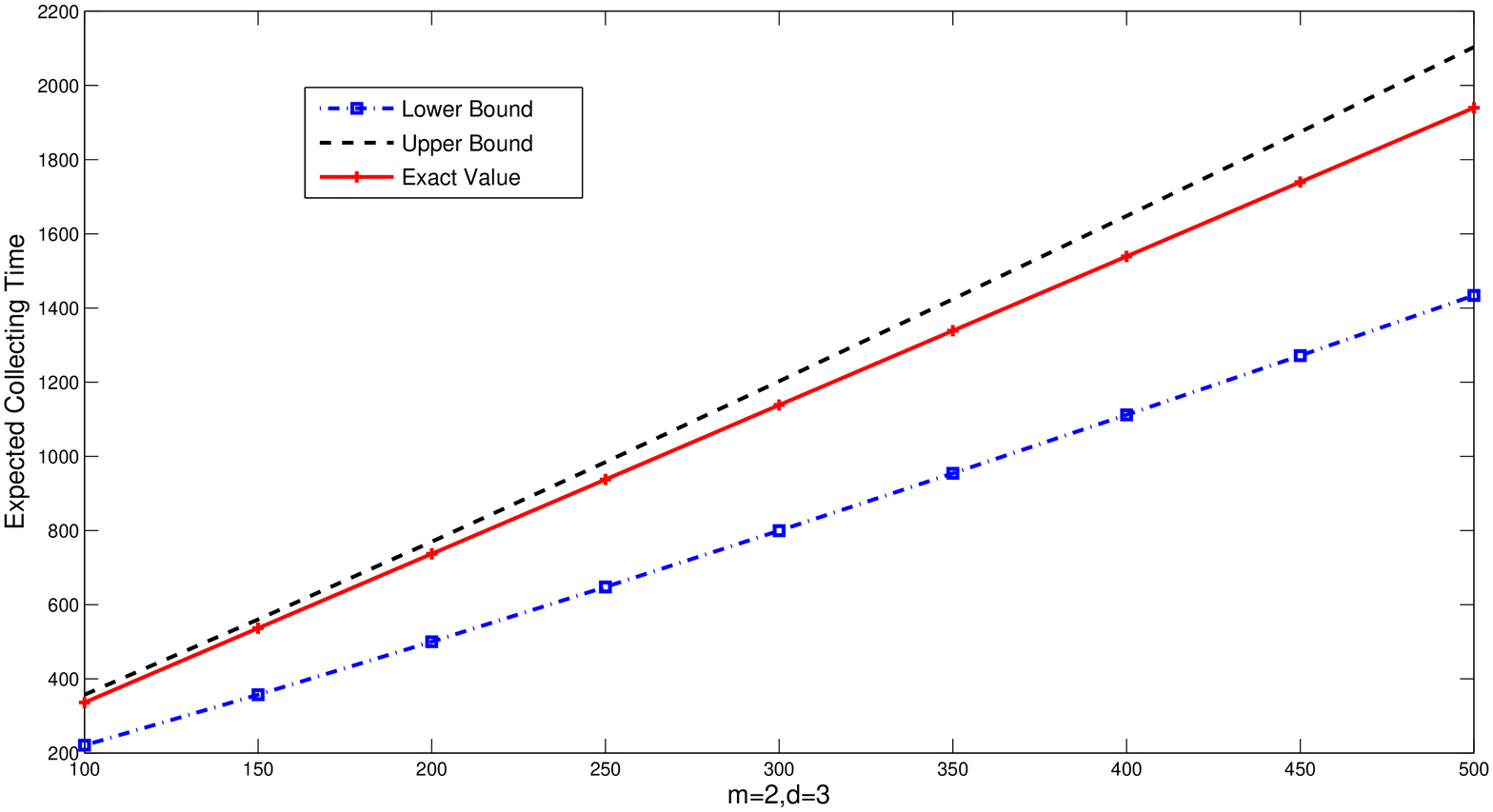}
\caption{Asymptotics for $m=2$, $d=3$.} \label{fig5}
\end{figure}

\section{Conclusion and Future Work}
In this paper, we considered a generalized coupon collector problem where the coupon collector needs to collect $m \geq 1$ sets of coupons and  has the freedom of keeping one coupon out of $d \geq 1$ coupons offered each time. We obtained asymptotically matching upper and lower bounds for the expected collection time.  We also provided an algorithm to calculate the expected collection time exactly based on a state representation for the coupon collecting process. We should note that asymptotically even if the coupon collector is only allowed to keep $1$ coupon out of the $d$ coupons, the needed time will still be shortened by a factor of $d$, as if the coupon collector is allowed to keep all the $d$ coupons offered each time.

There is much avenue for future work on this problem. First, one could attempt to get a closed-form expression for $E(D_{m,n}^d)$. Second, one could attempt to improve Algorithm 1. Algorithm 1 has a runtime of ${n+m \choose n}$.
To take advantage of this runtime requires constant time indexing. The direct approach is to index
the states in an $n$-dimensional matrix of size $(n+1)^m+1$. However, since there are a total of
${n+m \choose n}$ states, a large fraction of the matrix space is not required. Hence, it would be helpful
to find an algorithm which carries out triangular indexing in constant time. This would reduce the memory
requirements and increase the range of parameters over which the problem is computationally feasible. One could further observe that although there are ${n+m \choose n}$ states, only ${n+m-1 \choose m-1}$ are actually needed at any time. So with constant time triangular indexing, one can reduce the memory requirements further although the gain from the second reduction is
minimal.

\section*{Acknowledgment}
The research is supported by NSF under CCF-0835706. The authors would like to thank the input of Wuhan Desmond Cai.

\end{document}